\documentclass[11pt]{article}

\usepackage{url}
\usepackage{fullpage}
\usepackage{color}
\usepackage{graphicx}
\usepackage{amssymb}
\usepackage{amsfonts}
\usepackage{amsmath}
\usepackage{amsthm}

\usepackage{algorithm}
\usepackage{algpseudocode}

 \newtheorem{theorem}{Theorem}[section]
\newtheorem{lemma}[theorem]{Lemma}

\newtheorem{proposition}[theorem]{Proposition}

\theoremstyle{definition}

\theoremstyle{remark}

\newcommand{\IGNORE}[1]{}

\renewcommand{\H}{H}
\newcommand{\halo}{\mathrm{Halo}}
\newcommand{\cost}{c}

\renewcommand{\setminus}{-}
\newcommand{\opt}{opt}

\newcommand{\inv}[1]{#1^*} 
\newcommand{\Nbr}[1]{\Gamma(#1)}

\newcommand{\hC}{\widehat{C}}
\newcommand{\hG}{\widehat{G}}
\newcommand{\hE}{\widehat{E}}

\newcommand{\hr}{\widehat{r}}

\newcommand{\calA}{{\cal A}}

\newcommand{\setT}{T}
\newcommand{\setS}{S}
\newcommand{\setR}{R}
\newcommand{\nterms}{|T|}

\begin{document}

\title{An improved approximation algorithm for the minimum cost subset
  $k$-connected subgraph problem}
\author{Bundit Laekhanukit\thanks{
School of Computer Science, McGill University. 
\hbox{Email: blaekh@cs.mcgill.ca}.}
\thanks{
This research was supported by Natural Sciences and Engineering
Research Council of Canada (NSERC) grant no.~288334, 429598 and by
Harold H Helm fellowship.}\thanks{
Most of the work was done while the author was in Department of
Combinatorics and Optimization, University of Waterloo, Canada.}
}


\maketitle







\begin{abstract}
  The minimum cost subset $k$-connected subgraph problem is a
  cornerstone problem in the area of network design with vertex
  connectivity requirements. In this problem, we are given a graph
  $G=(V,E)$ with costs on edges and a set of terminals $T$. The goal 
  is to find a minimum cost subgraph such that every pair of terminals 
  are connected by $k$ openly (vertex) disjoint paths. 
  In this paper, we present an approximation algorithm for the subset
  $k$-connected subgraph problem which improves on the previous best
  approximation guarantee of $O(k^2\log{k})$ by Nutov (FOCS 2009). 
  Our approximation guarantee, $\alpha(\nterms)$, depends upon the
  number of terminals: 
\[
\alpha(\nterms) \ \ =\ \  
\begin{cases} 
O(k \log^2 k) & \mbox{if } 2k\le \nterms < k^2\\ 
O(k \log k) & \mbox{if } \nterms \ge k^2
\end{cases}
\]
  So, when the number of terminals is {\em large enough}, the  
  approximation guarantee improves significantly.
  Moreover, we show that, given an approximation algorithm for
  $\nterms=k$, we can obtain almost the same approximation guarantee
  for any instances with $\nterms> k$.
  This suggests that the hardest instances of the problem are when
  $\nterms\approx k$.

\end{abstract}




\section{Introduction}
\label{sec:intro}

We present an improved approximation algorithm for the {\em minimum
  cost subset $k$-connected subgraph problem}. In this problem (subset
$k$-connectivity, for short), we are given a graph $G=(V,E)$ with edge
costs and a set of terminals $\setT\subseteq V$. The goal is to find a 
minimum cost subgraph such that each pair of terminals is connected by
$k$ openly (vertex) disjoint paths.
This is a fundamental problem in network design which includes as
special cases the minimum cost Steiner tree problem (the case $k=1$)
and the minimum cost $k$ vertex-connected spanning subgraph problem
(the case $\setT=V$).
However, the subset $k$-connectivity problem is significantly harder
than these two special cases.
Specifically, an important result of Kortsarz, Krauthgamer and
Lee~\cite{KKL04} shows that the problem does not admit an
approximation guarantee better than $2^{\log^{1-\epsilon}{n}}$ for
any $\epsilon>0$ unless $\mathrm{NP}\subseteq
\mathrm{DTIME}(n^{O(\mathrm{polylog}(n))})$. 
In contrast, polylogarithmic approximation guarantees are known
for the minimum cost $k$-vertex connected spanning subgraph problem. 
The first such result was obtained by Fakcharoenphol and Laekhanukit
\cite{FL08} using 
the {\em Halo-set decomposition}, introduced by Kortsarz and
Nutov~\cite{KN05}. 
Subsequently, Nutov~\cite{Nutov09a} improved the approximation
guarantee to $O\left(\log k \cdot \log\frac{n}{n-k}\right)$.

Since the hardness result of Kortsarz et~at.~\cite{KKL04} no
non-trivial approximation algorithm was known for the general case of
the subset-$k$-connectivity problem until the work of Chakraborty,
Chuzhoy and Khanna~\cite{CCK08}.
They presented an $O(k^{O(k^2)}\cdot\log^4\nterms)$-approximation
algorithm for the rooted version of our problem, namely the 
{\em rooted subset $k$-connectivity problem}.
There, given a root vertex $r$ and a set of terminals $T$, 
the goal is to find a minimum cost subgraph that has $k$ openly
disjoint paths from the root vertex $r$ to every terminal in
$T$. Chakraborty et~al. showed how to solve the subset $k$-connectivity
problem by applying the rooted subset $k$-connectivity algorithm $k$
times, thus obtaining an 
$O(k^{O(k^2)}\cdot\log^4\nterms)$-approximation algorithm. 
Recently, in a series of developments 
\cite{CCK08,CK08,CK09,Nutov12,Nutov09c}, 
the approximation guarantees for the rooted subset $k$-connectivity
problem has been steadily improved. This has culminated in an
$O(k\log{k})$ guarantee due to Nutov~\cite{Nutov12},  
thus implying an approximation guarantee of $O(k^2\log{k})$ for the
subset $k$-connectivity problem. 


There is a trivial way to obtain an approximation bound of
$O(\nterms^2)$. 
So, with the current progress on the rooted subset $k$-connectivity
problem, the application of the rooted subroutine is only useful when
the number of terminals is large enough, say $\nterms\ge 2k$.
The main contribution of this paper is to show that, in this case,
only a polylogarithmic number of applications of the rooted subset
$k$-connectivity algorithm are required to solve the subset
$k$-connectivity problem.
Given an approximation algorithm for the rooted subset
$k$-connectivity problem, we show that only $O(\log^2{k})$
applications of the algorithm are required, and we can save a factor
of $O(\log{k})$ since some of these applications are applied to 
instances with lower costs. 
Moreover, as the number of terminal increases above $k^2$, we are able
to save an additional $O(\log k)$ factor.  
Thus, given an approximation algorithm for the rooted subset
$k$-connectivity problem in~\cite{Nutov12} due to Nutov (and with
careful analysis), we achieve an $\alpha(\nterms)$-approximation
guarantee where 
\[
\alpha(\nterms) \ \ =\ \   
\begin{cases} 
O(k \log^2 k) &\ \ \  \mbox{if } 2k\le \nterms < k^2\\ 
O(k \log k) &\ \ \  \mbox{if } \nterms \ge k^2
\end{cases}
\]


As we may combine our algorithm with the trivial
$O(\nterms^2)$-approximation algorithm for the case $\nterms<2k$, we
obtain an approximation guarantee of $O(k^2)$, which improves upon the
previous best approximation guarantee of $O(k^2\log{k})$ for all
cases.
Moreover, for $\nterms \geq 2k$, we obtain a significant improvement of a
factor of $k$.
Observe, however, that for the case $\nterms\approx k$ the guarantee
is still quadratic.  
At first, this may seem paradoxical since we may hope that the problem
is easier when the number of terminals is small. Our results suggest
that this is not the case.
Indeed, it appears that the hardest instances of subset
$k$-connectivity may have at most $k$ terminals.
Precisely, we show that, given an $\alpha(k)$-approximation algorithm
for the subset $k$-connectivity problem with $\nterms=k$, there is
an $(\alpha(k)+f(k))$-approximation algorithm for any instance with
$\nterms>k$, where $f(k)$ is the best known approximation guarantee
for the rooted subset $k$-connectivity problem.
Furthermore, we give an approximation preserving reduction from the
rooted subset $k$-connectivity problem to the subset $k$-connectivity
problem, showing a strong connection between the two problems.

\ \\
{\bf Related Work.} Some very special cases of the subset
$k$-connectivity problem are known to have constant factor approximation
algorithms.
For $k=1$, the minimum cost Steiner tree problem, the best known
approximation guarantee is $1.39$ due to
Byrka, Grandoni, Rothvo{\ss} and Sanit{\`a}~\cite{BGRS10}.
For $k=2$, a factor two approximation algorithm was given by
Fleischer, Jain and Williamson~\cite{FJW06}.
The subset $k$-connectivity problem also has an $O(1)$-approximation
algorithm when edge costs satisfy the triangle inequality; see
Cheriyan and Vetta~\cite{CV07}.
The most general problem in this area is the vertex-connectivity
survivable network design problem (VC-SNDP). In VC-SNDP, the
connectivity requirement for each pair of vertices can be arbitrary. 
Recently, Chuzhoy and Khanna~\cite{CK09} showed that there is an
$O(k^3\log n)$-approximation algorithm for VC-SNDP.
The problems where requirements are edge and element connectivity
(EC-SNDP and Element-SNDP) are also very well studied.
Both problems admit $2$-approximation algorithms via iterative 
rounding.
For EC-SNDP, a $2$-approximation algorithm was given by
Jain~\cite{Jain01}.
For Element-SNDP a $2$-approximation algorithm was given by Fleischer,
Jain and Williamson~\cite{FJW06}.
The vertex-cost versions of these problems have also been studied in
literature.  
Nutov~\cite{Nutov10-nodecost} gave an approximation guarantee of
$O(k\log\nterms)$ for vertex-cost EC-SNDP using a technique, called
spider decomposition. Later on, in~\cite{Nutov12}, Nutov applied
the spider decomposition technique to other vertex-cost problems,
giving approximation guarantees of $O(k\log\nterms)$ for Element-SNDP, 
$O(k^2\log\nterms)$ for the rooted subset $k$-connectivity problem,
$O(k^3\log\nterms)$ for the subset $k$-connectivity
problem and $O(k^4\log^2\nterms)$ for VC-SNDP.

\section{Preliminaries and Results}
\label{sec:prelim}

We begin with some formal definitions.
Let $G=(V,E)$ denote the graph for an instance of the problem. For a
set of edges $F$, the graph $G'=(V,E\cup F)$ is denoted by $G+F$; for
a vertex $v$, the graph obtained from $G$ by removing $v$ is denoted
by $G-v$.
For any set of vertices $U\subseteq V$, let $\Nbr{U}$ denote the set of
{\em neighbors} of $U$; that is, $\Nbr{U}=\{v\in V-U: \exists (u,v)\in E, u\in U\}$.
Define a set $\inv{U}$ to be $V\setminus(U\cup \Nbr{U})$, which is the  
{\em vertex-complement} of $U$. 
For any pair of vertices $s,t\in V$, two $s,t$-paths are
{\em openly disjoint} if they have no vertices except $s$ and $t$ in
common. 
Let $\setT\subseteq V$ be a set of vertices called {\em terminals}.
Without loss of generality, assume that no two terminals of $\setT$
are adjacent in $G$. 
This assumption can be easily justified by subdividing every edge
joining two terminals; that is, if there is an edge $(s,t)$ joining
two terminals, then we replace $(s,t)$ by two new edges $(s,u)$ and 
$(u,t)$ and set costs of the new edges so that $c(s,t)=c(s,u)+c(u,t)$,
where $c(.)$ is a cost function.
%
The set of terminals $T$ is {\em $k$-connected in $G$} if 
the graph $G$ has $k$ openly disjoint $s,t$-paths between every pair
of terminals $s,t\in \setT$.
Thus, by Menger Theorem, the removal of any set of
vertices with size at most $k-1$ leaves all the remaining terminals 
in the same component of the remaining graph.
By the {\em subset connectivity} of $G$ on $\setT$, we mean the
maximum integer $\ell$ such that $\setT$ is $\ell$-connected in $G$. 
A {\em deficient set} is a subset of vertices $U\subseteq V$ such that 
both $U$ and $\inv{U}$ contain terminals of $\setT$ and
$|\Nbr{U}|<k$. 
Observe that the 
vertex-complement $\inv{U}$ is also a deficient set.
Similarly, given a designated {\em root} vertex $r$, the graph is 
{\em $k$-connected from $r$ to $\setT$} if $G$ has $k$ openly
disjoint $r,t$-paths for every terminal $t\in \setT$ 
($r$ may or may not be in $\setT$).
By the {\em rooted connectivity} of $G$ from $r$ to $\setT$, we mean
the maximum integer $\ell$ such that $G$ is $\ell$-connected from $r$
to $\setT$. 

In the {\em subset $k$-connectivity problem}, we are given a graph
$G=(V,E)$ with a cost $\cost(e)$ on each edge $e\in E$, a set of
terminals $\setT\subseteq V$, and an integer $k\ge 0$. The goal is to
find a set of edges $\hE\subseteq E$ of minimum cost such that 
$\setT$ is $k$-connected in the subgraph $\hG=(V,\hE)$. 
In the {\em rooted subset $k$-connectivity problem}, our goal is to
find a set of edges $\hE\subseteq E$ of minimum cost such that the
subgraph $\hG=(V,\hE)$ is $k$-connected from $r$ to $\setT$, for a 
given root $r$.

Nutov~\cite{Nutov12} recently gave an $O(k\log{k})$-approximation
algorithm for the rooted subset $k$-connectivity problem.
The approximation guarantee improves by a logarithmic factor for the 
problem of increasing the rooted connectivity of a graph by one.

\begin{theorem}[Nutov 2009~\cite{Nutov12}]
  \label{thm:rooted-kconn}
  There is an $O(k\log{k})$-approximation algorithm for the rooted
  subset $k$-connectivity problem. Moreover, consider the restricted
  version of the problem where the goal is to increase the rooted
  connectivity from $\ell$ to $\ell+1$. 
  Then the approximation guarantee (with respect to a standard LP) is 
  $O(\ell)$.
\end{theorem}

Our focus is upon the subset $k$-connectivity problem.
%
The followings are our main results:

\begin{theorem}
  \label{thm:sskconn-main}
  For any set $T$ of terminals, there is an
  $\alpha(\nterms)$-approximation algorithm for the subset
  $k$-connectivity problem where
\[
\alpha(\nterms) \ \ =\ \   
\begin{cases} 
O(\nterms^2) &\ \ \ \mbox{if } \nterms < 2k 
                    \quad\mbox{(folklore)} \\ 
O(k \log^2 k) &\ \ \  \mbox{if } 2k\le \nterms < k^2\\ 
O(k \log k) &\ \ \  \mbox{if } \nterms \ge k^2
\end{cases}
\]
In particular, there is an $O(k^2)$-approximation algorithm for the
general case of the subset $k$-connectivity problem, and there is an
$O(k\log{k})$-approximation algorithm when $\nterms \ge k^2$.
\end{theorem}

\begin{proposition}
\label{prop:hardest}
Consider the subset $k$-connectivity problem.
Suppose there is an $\alpha(k)$-approximation algorithm for instances
with $\nterms=k$.
Then there is an $(\alpha(k)+f(k))$-approximation algorithm for
any instance with $\nterms>k$, where $f(k)$ is the best known
approximation guarantee for the rooted subset $k$-connectivity problem.  
\end{proposition}


\begin{theorem}
\label{thm:rooted-to-subset}
There is an approximation preserving reduction such that, given an
instance of the rooted subset $k$-connectivity problem consisting of a
graph $G$, a root vertex $r$ and a set of terminals $T$, outputs an
instance of the subset $k$-connectivity problem consisting of a graph
$G'$ and a set of terminals $T\cup\{r\}$. 
\end{theorem}

The hardness result in Theorem~\ref{thm:rooted-to-subset} together
with the hardness of the rooted subset $k$-connectivity problem by
Cheriyan, Laekhanukit, Naves and Vetta~\cite{CLNV12} implies the
hardness of $\Omega(k^\epsilon)$, for the subset $k$-connectivity
problem, where $\epsilon>0$ is some fixed constant.

Some results and proofs similar to the ones in this paper have
appeared in previous literature; see~\cite{CVV03,KN05,Bundit-Thesis}.  
In particular, Lemma~\ref{lmm:low-thickness} and
Lemma~\ref{lmm:two-cores} appeared in~\cite{KN05}
and~\cite{Bundit-Thesis}, respectively. 
The proofs of Proposition~\ref{prop:hardest} is identical to that of
the case $\setT=V$, which was given in \cite{KR96} and also in
\cite{ADNP99}.
Our key new contributions are
Lemmas~\ref{lmm:bound-1},~\ref{lmm:halo-nbr} and~\ref{lmm:num-cores},
which allow us to extend the result in~\cite{KN05} to the subset
$k$-connectivity problem.

We remark that, at the time this paper is written, the approximation
guarantee of the subset $k$-connectivity problem was improved by
Nutov~\cite{Nutov11} to $O(k\log{k})$ for all $k\leq |T|-o(|T|)$.

\bigskip

\noindent{\bf Organization: }
In Section~\ref{sec:algo}, we present an approximation algorithm for
the subset $k$-connectivity problem, which is the main result in this
paper. 
In Section~\ref{sec:below2k}, we give a discussion that our algorithm
and analysis can be extended to the case $k<\nterms< 2k$. 
To keep the presentation simple, Section~\ref{sec:below2k} is
presented separately from the main result.
%
%
In Section~\ref{sec:hardness}, we discuss the hardness of the subset
$k$-connectivity problem. To be precise, we show that the hardest
instance of the subset $k$-connectivity problem might be when
$|T|\approx{k}$, and we give an approximation preserving reduction
from the rooted subset $k$-connectivity problem to the subset
$k$-connectivity problem.
\\

\section{An approximation algorithm}
\label{sec:algo}

Our main result in Theorem \ref{thm:sskconn-main} breaks up into three
cases where there are a small number, a moderate number and a large
number of terminals, respectively. 
Indeed, the first case is a folklore. 
When there are a small number of terminals ($\nterms<2k$), we apply
the following trivial $O(\nterms^2)$-approximation algorithm. We find
$k$ openly disjoint paths of minimum cost between every pair of
terminals by applying a minimum cost flow algorithm. Let $\opt$
denote the cost of the optimal solution to the subset $k$-connectivity
problem. Since any feasible solution to the subset $k$-connectivity
problem has $k$ openly disjoint paths between every pair of terminals, 
the cost incurred by finding a minimum cost collection of $k$
openly disjoint paths between any pair of terminals is at most
$\opt$. Since we have at most $\nterms^2$ pairs, this incurs a total
cost of $O(\nterms^2\cdot\opt)$. 

The remaining two cases are similar. Things are slightly easier,
though, when there are large number of terminals ($\nterms \ge k^2$),
leading to a slightly better guarantee than when there are a moderate
number of terminals ($2k \le \nterms$). We devote most of this section to
presenting an approximation algorithm for the moderate case. (In
Section~\ref{sec:very-simple-algo}, we show the improvement
for the case of a large number of terminals.)


Our algorithm works by repeatedly increasing the subset connectivity
of a graph by one. 
We start with a graph that has no edges.
Then we apply $k$ {\em outer} iterations. 
Each outer iteration increases the subset
connectivity (of the current graph) by one by adding a set of edges of
approximately minimum cost. The analysis of the outer iterations
applies linear programming (LP) scaling and incurs a factor of
$O(\log{k})$ in the approximation guarantee for the $k$ outer
iterations.
%
The analysis based on LP-scaling can be seen
in~\cite{RW97,CV07,KN05,FL08} and also
in~\cite{GGPSTW94,Bundit-Thesis}. 

The following is a standard LP-relaxation for the subset
$k$-connectivity problem. 
\[
\begin{array}{lll}
\min & \displaystyle\sum_{e\in E}c_ex_e \\ 
\mbox{s.t.} 
  &
    \displaystyle\sum_{e\in\delta(U,W)}x_e \geq k - |V\setminus(U\cup W)|  
    & \forall (U,W)\in\mathcal{S}\\
  & 0 \leq x_e \leq 1 & \forall e\in E
\end{array}
\]
where $\delta(U,W)=\{(u,w)\in E:u\in U, w\in W\}$ is a set of edges
with one endpoint in $U$ and the other endpoint in $W$, and 
$\mathcal{S}=\{(U,W)\in V\times V: U\cap W=\emptyset,
U\cap{T}\neq\emptyset, W\cap{T}\neq\emptyset\}$.

\begin{lemma}
Suppose there is a $\beta(\ell)$-approximation algorithm for the
problem of increasing the subset connectivity of a graph from $\ell$
to $\ell+1$ with respect to a standard LP, where $\beta(\ell)$ is a
non-decreasing function.Then there is an $O(\beta(k)\log
k)$-approximation algorithm for the subset $k$-connectivity problem.
\label{lmm:LP-scaling}
\end{lemma}

We are left with the key problem of increasing the subset connectivity
(of the current graph) by one by adding a set of edges of
approximately minimum cost. Throughout this section, we assume that
the set of terminals $T$ is $\ell$-connected in the current graph, and 
$\nterms\geq 2k \geq 2\ell$. 
Also, we assume that no two terminals are adjacent in the input graph
$G=(V,E)$.
 \\

\noindent{\bf Assumption:} 
The set of terminals $T$ is $\ell$-connected in the current graph
$\hG=(V,\hE)$. 
Moreover, no two terminals are adjacent in the input graph $G$.
\\ 

Our algorithm solves the problem of increasing the subset connectivity
of a graph by one by applying a number of so-called {\em inner}
iterations.
To describe our algorithm, we need some definitions and subroutines.
Thus, we defer the description of our algorithm to
Section~\ref{sec:second-algo}.
In Section~\ref{sec:deficient}, we give important definitions and
structures of subset $\ell$-connected graphs called ``cores'' and
``halo-families''. 
Our algorithm requires two subroutines. 
The first one is the subroutine that employs the rooted subset
$(\ell+1)$-connectivity algorithm to cover halo-families. This
subroutine is given in Section~\ref{sec:cover-halo}. 
The second one is the subroutine for decreasing the number of cores to
$O(\ell)$, which is given in Section~\ref{sec:reduce-cores}.
Then we introduce a notion of ``thickness'' in
Section~\ref{sec:thickness}.  
This notion guides us how to use the rooted subset
$(\ell+1)$-connectivity algorithm efficiently. 
Finally, in Section~\ref{sec:second-algo}, we present an
$O(k\log^2{k})$-approximation algorithm for the case $\nterms\ge2k$. 
By slightly modifying the algorithm and analysis, we show in
Section~\ref{sec:very-simple-algo} that our algorithm achieves a
better approximation guarantee of $O(k\log{k})$ when $\nterms\ge
k^2$.


\subsection{Subset $\ell$-connected graphs:
	deficient sets, cores, halo-families and halo-sets}
\label{sec:deficient}
\label{sec:core-halo}

In this section, we discuss some key properties of deficient sets
that will be exploited by our approximation algorithm. 

Assume that the set of terminals $\setT$ is $\ell$-connected
in the graph $G=(V,E)$.
Then $G$ has $|\Nbr{U}|\geq\ell$ for
all $U\subseteq V$ such that $U\cap \setT\neq\emptyset$
and $\inv{U}\cap \setT\neq\emptyset$. 
Moreover, by Menger Theorem,
$G$ is subset $(\ell+1)$-connected
if and only if $G$ has no deficient set.

A key property of vertex neighborhoods is that
the function $|\Nbr{\cdot}|$ on subsets of $V$ is submodular.
In other words,
for any subsets of vertices $U,W\subseteq V$, 
\[
|\Nbr{U\cup W}|+|\Nbr{U\cap W}| \leq |\Nbr{U}| + |\Nbr{W}|.
\]

We call a deficient set $U\subseteq V$ {\em small}\footnote{
There is another way to define a {\em small} deficient set. For
example, in~\cite{Nutov11}, Nutov defined a small deficient set as a
deficient set $U$ such that $|U\cap\setT|\leq\frac{\nterms-\ell}{2}$. 
}
if $|U\cap\setT| \le |\inv{U}\cap\setT|$.

\begin{proposition}
For any small deficient set $U$, $|U\cap \setT|\leq \nterms/2$ and
$|\inv{U}\cap \setT|\geq (\nterms-\ell)/2$.
\label{prop:size}
\end{proposition}

\begin{proof}
The first inequality follows from the definition of small deficient
sets. Consider the second inequality. We have
\[
|\inv{U}\cap \setT|\geq
\frac{|U\cap \setT|+|\inv{U}\cap \setT|}{2}=
\frac{\nterms-|\Nbr{U}\cap \setT|}{2} \geq \frac{\nterms-\ell}{2}
\]
\end{proof}

\begin{lemma}[Uncrossing Lemma]
\label{lmm:uncross}
Consider any two distinct deficient sets $U,W\subseteq V$. If 
$U\cap W\cap \setT\neq\emptyset$ and $\inv{U}\cap\inv{W}\cap
\setT\neq\emptyset$, then both $U\cap W$ and $U\cup W$ are
deficient sets.
Moreover, if $U$ or $W$ is a small deficient set, then
$U\cap W$ is a small deficient set.
\end{lemma}

\begin{proof}

Suppose $U\cap W\cap \setT\neq\emptyset$ and $\inv{U}\cap\inv{W}\cap
\setT\neq\emptyset$. Note that 
\begin{align*}
\inv{U}\cap\inv{W}\cap \setT 
 &= (V-(U\cup \Nbr{U}))\cap (V-(W\cup \Nbr{W}))\cap \setT \\
 &= (V-(U\cup W\cup \Nbr{U}\cup \Nbr{W})) \cap \setT \\
 &= (V-((U\cup W)\cup \Nbr{U\cup W})) \cap \setT \\
 &= \inv{(U\cup W)}\cap \setT.
\end{align*}

Moreover, $\inv{(U\cup W)}\subseteq\inv{(U\cap W)}$. This means that
\[
(U\cap W)\cap \setT\neq\emptyset,\quad
(U\cup W)\cap \setT\neq\emptyset,\quad 
\inv{(U\cup W)}\cap \setT\neq\emptyset\quad \mbox{and} \quad  
\inv{(U\cap W)}\cap \setT\neq\emptyset.
\]

Hence, by Menger Theorem, we have $|\Nbr{U\cup W}|\geq\ell$
and $|\Nbr{U\cap W}|\geq\ell$. Moreover, since $U,W$ are deficient
sets, we have $|\Nbr{U}|=|\Nbr{W}|=\ell$. It then follows by the 
submodularity of $|\Nbr{.}|$ that
\[
2\ell\leq |\Nbr{U\cup W}|+|\Nbr{U\cap W}| \leq |\Nbr{U}| +
|\Nbr{W}|=2\ell. 
\]

Thus, $|\Nbr{U\cap W}|=|\Nbr{U\cup W}|=\ell$. This implies that both 
$U\cup W$ and $U\cap W$ are deficient sets. Moreover, suppose $U$ or
$W$ is a small deficient set. Without loss of generality, assume 
that $U$ is a small deficient set. Then $U\cap W$ is a small deficient
set because $U\cap W\subseteq U$. Thus,
\[
|U\cap W\cap \setT|\leq |U\cap \setT|\leq |\inv{U}\cap \setT| 
  \leq |\inv{(U\cap W)}\cap \setT| 
\]
\end{proof}


By a {\em core}, we mean a small deficient set $C$ that is
inclusionwise minimal. In other words, $C$ is a core if it
is a small deficient set that does not contain another such set. 
It can be seen that any small deficient set $U$ contains at least one
core. 

The {\em halo-family} of a core $C$,
denoted by $\halo(C)$, is
the set of all small deficient sets that contain
$C$ and contain no other cores; that is,
\[
\halo(C)=\{U: \mbox{$U$ is a small deficient set,
	$C\subseteq U$, and
	there is no core $D\neq C$ such that $D\subseteq U$}\}.
\]

%
The {\em halo-set} of a core $C$, denoted by $H(C)$,
is the union of all the sets in $\halo(C)$; that is, 

\[
H(C) = \bigcup\{U:U\in\halo(C)\}
\]

An example of cores, halo-families and halo-sets is illustrated in 
Figure~\ref{fig:core-haloset}.

\begin{figure}[hbt]
\fbox{\begin{minipage}{\textwidth}
     \centerline{ \includegraphics[scale=0.6]  {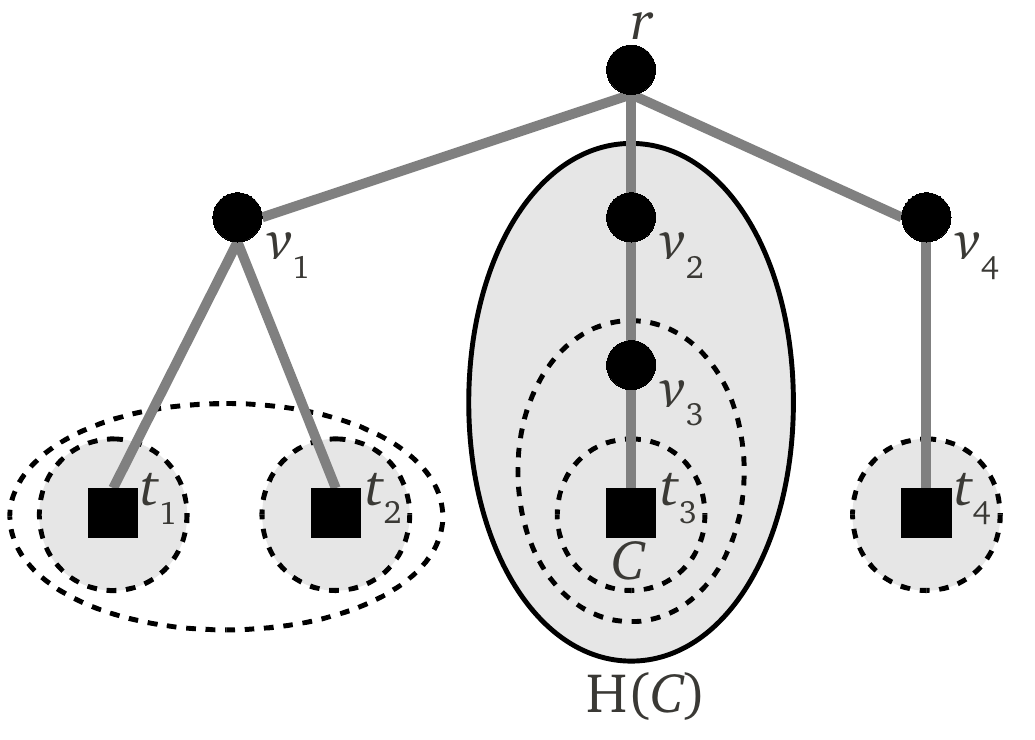} }
\caption[An example of cores, halo-families and halo-sets]
{
The figure shows an example of cores, halo-families
and halo-sets for $\ell=1$. In the figure, we are 
given a graph $\hG=(V,\hE)$. The graph $\hG$ is a tree constructed by
adding a vertex $r$ as a root. Then we add terminals
$t_1,t_2,t_3,t_4$ and add paths $(t_1,v_1,r)$, $(t_2,v_1,r)$,
$(t_3,v_3,v_2,r)$ and $(t_4,v_4,r)$ connecting terminals to the
root vertex. Clearly, $\setT=\{t_1,t_2,t_3,t_4\}$ is $\ell$-connected
in $\hG$. Every leaf vertex of the tree $\hG$ forms a core.
For example, $C=\{t_3\}$ is a core. The halo-family of
$C$ is $\halo(C)=\{\{t_3\},\{t_3,v_3\},\{t_3,v_2,v_3\}\}$, and the
halo-set of $C$ is $\H(C)=\{t_3,v_2,v_3\}$. Observe that the set 
$\{t_1,t_2\}$ is a small deficient set, but it does not belong to any
halo-family. 
}
\label{fig:core-haloset}
\end{minipage}}
\end{figure}

\ \\ {\bf Remark}: 
We remark that cores and halo-sets of subset $\ell$-connected
graphs can be computed in polynomial time. In fact, algorithms for 
computing cores and halo-sets of the $k$-vertex connected spanning
subgraph problem also apply to the subset $k$-connectivity problem. 
See~\cite{KN05,FL08,Bundit-Thesis}. \\

Some important properties of cores and halo-families that we will
require are stated below. 

\begin{lemma}[Disjointness Lemma]
\label{lmm:sskconn-disjoint}
Consider any two distinct cores $C$ and $D$. For any deficient sets 
$U\in\halo(C)$ and $W\in\halo(D)$, either 
$U\cap W\cap \setT=\emptyset$ or $\inv{U}\cap\inv{W}\cap \setT=\emptyset$. 
\end{lemma}

\begin{proof}
Suppose to the contrary that
$U\cap W\cap \setT\neq\emptyset$ and
$\inv{U}\cap\inv{W}\cap \setT \neq\emptyset$.
Then, by Lemma~\ref{lmm:uncross}, $U\cap W$ is a small deficient set.
Thus, $U\cap W$ contains a core. This core is either $C$ or $D$ or
another core distinct from $C$ and $D$. In each case, we have a
contradiction.
\end{proof}

The next result gives an upper bound on the number of halo-sets
that contain a chosen terminal, which is a key for the design of our
algorithm. 

\begin{lemma}[Upper bound]
\label{lmm:bound-1}
For any terminal $t\in \setT$, 
the number of cores $C$ such that $t\in H(C)$ is at most 
$\frac{2(\nterms-1)}{\nterms-\ell}$. 
\end{lemma}

\begin{proof}
Let $C_1,C_2,\ldots,C_q$ be distinct
cores such that $t\in H(C_i)$ for all $i=1,2,\ldots,q$.
For each $i=1,2,\ldots,q$, since $t$ is in the halo-set $H(C_i)$, there
must exist a deficient set $U_i$ in the halo-family $\halo(C_i)$
that contains $t$. It then follows that $t\in\bigcap_{i=1}^qU_i$. By
the Disjointness Lemma (Lemma~\ref{lmm:sskconn-disjoint}), for $i\neq j$,
$U_i\cap U_j\cap \setT\neq\emptyset$ only if $\inv{U_i}\cap\inv{U_j}\cap 
\setT=\emptyset$. This is because $U_i$ and $U_j$ are deficient sets of
different halo-families. For $i=1,2,\ldots,q$, observe that
$U_i$ is small. Hence, by Proposition~\ref{prop:size},
we have $|\inv{U_i}\cap \setT|\geq(\nterms-\ell)/2$.
Thus, the upper bound on the number halo-sets
that contain $t$ is
\[
\frac{\nterms-1}{(\nterms-\ell)/2}=\frac{2(\nterms-1)}{\nterms-\ell}
\]
\end{proof}


\subsection{Covering halo-families via rooted subset
  $(\ell+1)$-connectivity} 
\label{sec:cover-halo}

We say that an edge $e=(u,v)$ 
{\em covers} a deficient set $U$ if $e$ connects $U$ and $\inv{U}$; 
that is, $u\in U$ and $v\in\inv{U}$.
Clearly, $e$ covers $U$ if $e$ covers $\inv{U}$. 
Observe that if $e$ covers $U$,
then after adding the edge $e$ to the current graph,
$U$ is no longer a deficient set.
Now, consider any core $C$. We say that a set of edges $F$ 
{\em covers} the halo-family of $C$ if each deficient set $U$ in
$\halo(C)$ is covered by some edge of $F$.
For a terminal $r\in T$, we say that the terminal $r$ {\em hits} the
halo-family $\halo(C)$ if $r$ is in $C$ or $r$ is in the
vertex-complement of the halo-set of $C$; that is, $r$ hits
$\halo(C)$ if $r\in C$ or $r\in\H(C)^*$.
For a set of terminals $\setS\subseteq\setT$, we say that $\setS$ 
{\em hits} a halo-family $\halo(C)$ if there is a terminal $r\in\setS$
that hits $\halo(C)$. 
The following lemma shows that if $r$ hits the halo-family
$\halo(C)$, then we can find a set of edges $F$ that covers $\halo(C)$
by applying the rooted subset $(\ell+1)$-connectivity algorithm with
$r$ as the root.

\begin{lemma}
\label{lmm:cover-halo}
Consider a set of edges $F$ whose addition to $\hG$ makes
the resulting graph $\hG+F$ $(\ell+1)$-connected 
from a terminal $r$ to $\setT$.
Let $C$ be any core.
If $r\in C$ or $r\in\inv{H(C)}$, then
$F$ covers all deficient sets in the halo-family of $C$. 
\end{lemma}

\begin{proof}
Consider the graph $\hG+F$. By the construction, $\hG+F$ has
$(\ell+1)$ openly disjoint paths from $r$ to every terminal of
$\setT$.
This means that $F$ covers all deficient sets of $\hG$ that contains
$r$. 
If $r\in C$, then $r\in U$ for all deficient sets $U\in\halo(C)$. 
So, $F$ covers $\halo(C)$.
Similarly,
If $r\in \inv{H(C)}$, then $r\in\inv{U}$ for all deficient sets
$U\in \halo(C)$. 
So, again, $F$ covers $\halo(C)$ because an edge $e\in F$ covers $\inv{U}$
if and only if $e$ covers $U$, and the lemma follows.
\end{proof}


\subsection{Preprocessing to decrease the number of cores}
\label{sec:reduce-cores}

In this section, we describe the preprocessing algorithm that
decreases the number of cores to
$O\left(\frac{\ell \nterms}{\nterms-\ell}\right)$.
We apply the following {\em root padding algorithm} in the
preprocessing step.  

\ \\{\bf The root padding algorithm: }
The algorithm takes as an input a graph $G=(V,E)$ with the given edge
costs, a subset of terminals $\setR\subseteq\setT$, and a connectivity
parameter $\rho\leq|\setR|$. 
We construct a padded graph by adding a new vertex $\hr$ and new edges
of zero cost from $\hr$ to each terminal of $\setR$. Then we apply
the rooted subset $\rho$-connectivity algorithm to the padded graph
with the set of terminals $\setT$ and the root $\hr$.
We denote a solution subgraph (of the padded graph) by 
$\hG = (V\cup\{\hr\},\,F\cup\{(\hr,t):t\in \setR\})$, where
$F\subseteq E$. 
Then the algorithm outputs the subgraph (of the original graph)
$\hG-\hr=(V,F)$. 
The following result shows that, in the resulting graph
$\hG$, every deficient set contains at least one terminal of
$\setR$.

\begin{lemma}[root padding]
  \label{lem:rootpad}
  Suppose we apply the root padding algorithm as above, and it finds a
  subgraph $\hG-\hr=(V,F)$.
  Then every deficient set of $\hG-\hr$ (with respect to
  $\rho$-connectivity of $\setT$ in $\hG-\hr$)
  contains at least one terminal of $\setR$.
\end{lemma}

\begin{proof}
  Observe that $\hG$ has $\rho$ openly disjoint $\hr,t$-paths,
  for all $t\in\setT$.  Suppose $U\subseteq V$ is a deficient set of
  $\hG-\hr$ that contains none of terminals of $\setR$. Thus, $U$
  contains another terminal $t\in T\setminus \setR$ and $|\Nbr{U}|<\rho$.
  Then $(\hG-\hr)-\Nbr{U}$ has no path between $t$ and a terminal
  of $\setR$.  This also holds for $\hG-\Nbr{U}$ because adding
  $\hr$ and the edges from $\hr$ to every terminals of $\setR$
  cannot give a path between $t$ and a terminal of $\setR$.  This is a
  contradiction since $\hG$ should have $\rho$ openly disjoint
  $\hr,t$-paths.
\end{proof}

\noindent {\bf Remark}:
Consider an instance $\Pi_{root}$ of the rooted subset $k$-connectivity
problem that is obtained from an instance $\Pi_{subset}$ of the subset 
$k$-connectivity problem by either
~(1) picking one of the terminals as the root or
~(2) applying the root padding algorithm using any $k$ terminals of
$\Pi_{subset}$.
Then the cost of the optimal solution to $\Pi_{root}$ is at most the
cost of the optimal solution to $\Pi_{subset}$ because any feasible
solution to $\Pi_{subset}$ gives a feasible solution to $\Pi_{root}$
(but not vice-versa). \\

Next, recall that $\nterms\geq 2\ell$. 
We apply the root padding algorithm in Lemma~\ref{lem:rootpad} 
to any subset $\setR$ of $(\ell+1)$ terminals with $\rho=(\ell+1)$.
By Theorem~\ref{thm:rooted-kconn}, this incurs a cost of at most
$O((k\log{k})\cdot\opt)$ 
Moreover, the algorithm adds a set of edges
to the current graph such that
every deficient set of the resulting graph
contains at least one terminal of $\setR$.
Thus, each core of the resulting graph
contains at least one terminal of $\setR$. 
By Lemma~\ref{lmm:bound-1}, each terminal is in
$O\left(\frac{\nterms}{\nterms-\ell}\right)=O(1)$ halo-sets.
Hence, the number of cores in the resulting graph is at most 
$O(\ell)$.
This gives the next result.

\begin{lemma}
\label{lmm:num-cores}
Given a subset $\ell$-connected graph, where $\nterms\geq 2\ell$, 
there is an $f(k)$-approximation algorithm that decreases the number
of cores to $O(\ell)$, where $f(k)$ is the best known approximation
guarantee for the rooted subset $k$-connectivity problem.
\end{lemma}


\subsection{Thickness of terminals}
\label{sec:thickness}
Consider a graph $\hG$ such that $\setT$ is $\ell$-connected in $\hG$.
We define the {\em thickness} of a terminal $t\in \setT$ to be the
number of halo-families $\halo(C)$ such that $t\in \Nbr{H(C)}$. 
Thus, the thickness of a terminal $t$ is 
$|\{\halo(C):C\mbox{ is a core}, t \in \Nbr{H(C)}\}|$.

The following lemmas show the existence of a terminal with low thickness. 

\begin{lemma}
\label{lmm:halo-nbr}
For every core $C$, $|\Nbr{H(C)}|\leq \ell$.
\end{lemma}

\begin{proof}
We use induction on the number of deficient sets in $\halo(C)$.
For the induction basis, $\halo(C)$ has one deficient set $U$,
and $H(C)=U$.
Then $\Nbr{H(C)} = \Nbr{U}$ has size $\ell$ since the graph
is subset $\ell$-connected.

Suppose that $\mathcal{U}$ is the union of $j$ deficient sets that
each contains the core $C$, and suppose that
$|\Nbr{\mathcal{U}}|\le\ell$. 
Consider another deficient set $W$ that contains $C$.
Our goal is to show that $\Nbr{\mathcal{U}\cup{W}}$ has size at most
$\ell$. 
If $W\subseteq\mathcal{U}$, then we are done.
Otherwise, we apply the submodularity of $|\Nbr{\cdot}|$.
Observe that $\mathcal{U}\cap{W}$ contains a terminal since 
$\mathcal{U}\cap{W}\supseteq{C}$, and $\inv{(\mathcal{U}\cap{W})}$
contains a terminal since 
$\inv{(\mathcal{U}\cap W)}\supseteq\inv{W}$ and $\inv{W}$ contains a
terminal. 
Hence, $\Nbr{\mathcal{U}\cap{W}}$ has size at least $\ell$.
Thus, we have  
\[ 2\ell \ge |\Nbr{\mathcal{U}}| + |\Nbr{W}| \ge
	|\Nbr{\mathcal{U}\cup{W}}| + |\Nbr{\mathcal{U}\cap{W}}| \ge 
	|\Nbr{\mathcal{U}\cup{W}}| + \ell.
\]
This implies that $\Nbr{\mathcal{U}\cup{W}}$ has size at most $\ell$,
and the lemma follows.
\end{proof}

The following lemma shows the existence of a terminal with low
thickness. 

\begin{lemma}
\label{lmm:low-thickness}
Consider a subset $\ell$-connected graph.
Let $q$ denote the number of halo-families.
Then there exists a terminal $t\in \setT$ with thickness
at most $\frac{\ell q}{\nterms}$.
\end{lemma}

\begin{proof}
Consider the following bipartite incidence graph $B$
of terminals and halo-families:
$B$ has a vertex for each terminal and each halo-family,
and it has an edge between a terminal $t$ and a halo-family $\halo(C)$
if and only if $t\in\Nbr{H(C)}$.
The previous lemma shows that
each halo-family is adjacent to at most $\ell$ terminals in $B$.
Hence, $B$ has at most $\ell q$ edges.
Therefore, $B$ has a terminal that is adjacent to
at most $\frac{\ell q}{\nterms}$ cores; that is,
there exists a terminal with the required thickness.
\end{proof}


\subsection{An $O(k\log^2{k})$-approximation algorithm for
  $\nterms\geq 2k$} 
\label{sec:second-algo}

In this section, we describe our approximation algorithm for the case
of a moderate number of terminals. Recall that we solve the problem by
iteratively increasing the subset connectivity of a graph by one. 
Initially, we apply the algorithm in Section~\ref{sec:reduce-cores} to
decrease the number of core to $O(\ell)$. 
Then we apply inner iterations until all the deficient sets are
covered. 
At the beginning of each inner iteration, we compute the cores
and the halo-sets.
Then we apply a {\em covering-procedure} to find a set of edges that 
covers all the computed halo-families. 
This completes one inner iteration.
Note that an inner iteration may not cover all of the deficient sets 
because deficient sets that contain two or more of the initial cores
(those computed at the start of the inner iteration)
may not be covered.
So, we have to repeatedly apply inner iterations until no core is
present.
%
See Algorithm~\ref{algo:2nd-algo}.


\begin{algorithm}
\caption{An approximation algorithm for moderate and large number of
terminals}
\label{algo:2nd-algo}
\begin{algorithmic}
\For {$\ell=0,1,\ldots,k-1$} \Comment{(outer iterations)}
  \State {\em (* Increase the subset connectivity of a graph by one. *)}
  \State Decrease the number of cores to $O(\ell)$.
  \While {the number of cores is greater than $0$} 
     \Comment{(inner iterations)} 
     \State Compute cores and halo-sets.
     \State Apply a {\bf covering-procedure} to cover all the halo-families
  \EndWhile 
\EndFor
\end{algorithmic}
\end{algorithm}

We now describe the covering-procedure.
The procedure first finds a set of terminals $\setS\subseteq\setT$
that hits all the computed halo-families.
Then it applies the rooted subset $(\ell+1)$-connectivity algorithm 
(Theorem~\ref{thm:rooted-kconn}) from each terminal of $\setS$. 
Let $F$ be the union of all edges found by the rooted subset 
$(\ell+1)$-connectivity algorithm.
Then, by Lemma~\ref{lmm:cover-halo}, $F$ covers all the
halo-families.

The key idea of our algorithm is to pick a terminal $\hr$ with a
minimum thickness. Observe that a halo-family $\halo(C)$ is not hit by 
$\hr$ only if 
\begin{itemize}
\item[(1)] its halo-set $\H(C)$ has $\hr$ as a neighbor (that is,
  $\hr\in \Nbr{H(C)}$) or 
\item[(2)] its halo-set $H(C)$ contains $\hr$, but its core $C$ does
  not contain $\hr$.
\end{itemize}
The number of halo-families $\halo(C)$ such that $\hr\in \Nbr{H(C)}$
may be large, but the number of halo-families whose halo-sets contain 
$\hr$ is $O(1)$, assuming that $\nterms\ge 2\ell$.
Hence, we only hit halo-families of the second case by picking one
terminal from each core $C$ whose halo-set contains $\hr$.
Thus, the number of terminals picked is $O(1)$. 
We call this a {\em micro} iteration.
Then the remaining halo-families are the
halo-families whose halo-sets have $\hr$ as a neighbor.
We repeatedly apply micro iterations until
we hit all of the halo-families computed
at the start of the inner iteration.

To be precise, initially let $\setS=\emptyset$. 
In each micro iteration, we add to $\setS$ a terminal $\hr$ of minimum
thickness (with respect to halo-families that are not hit by $\setS$).  
Then, for each core $C$ such that $\hr\in H(C)\setminus C$,
we add to $\setS$ any terminal in $C\cap\setT$.
We repeatedly apply micro iterations until $\setS$ hits all the
halo-families.  
At the termination, we apply  the rooted subset
$(\ell+1)$-connectivity algorithm  (Theorem~\ref{thm:rooted-kconn})
from each terminal of $\setS$, and we return all the set of edges
found by the algorithm as an output. 
%
The covering-procedure is presented in Figure~\ref{algo:cover-proc}.


\begin{algorithm}
\caption{Covering-procedure}
\label{algo:cover-proc}
\begin{algorithmic}
\State $\setS\leftarrow\emptyset$.
\While {some halo-family is not hit by $\setS$} 
  \Comment{(micro iteration)}
  \State Add to $\setS$ a terminal $\hr$ with a minimum thickness.
  \For {each halo-family $\halo(C)$ (not hit by $\setS$) such
  that $\hr\in H(C)\setminus C$}
    \State  Add to $\setS$ any terminal $r\in C$.
  \EndFor
\EndWhile
\For {each terminals $r$ in $\setS$}
  \State Apply the rooted subset $(\ell+1)$-connectivity algorithm
  from $r$. 
\EndFor
\end{algorithmic}
\end{algorithm}

\subsubsection{Analysis}
\label{sec:analysis}

The feasibility of a solution directly follows from the
condition of the inner iteration; that is, the inner iteration
terminates when a current graph has no core. 
So, at the termination of the inner iteration, the resulting graph has
no deficient set. 
Thus, the subset connectivity of the graph becomes $\ell+1$. 
Applying the outer iteration $k$ times, the final graph is then subset 
$k$-connected.

It remains to analyze the cost of the solution subgraph. 
First, we analyze the number of times that the covering-procedure
applies the rooted subset $(\ell+1)$-connectivity algorithm. 
Then we analyze the total cost incurred by all inner iterations, which
is the cost for increasing the subset connectivity of a graph by one. 
Finally, we apply Theorem~\ref{lmm:LP-scaling} to analyze the final
approximation guarantee.  

Consider any micro iteration of the covering-procedure.
By Lemma~\ref{lmm:bound-1}, $\hr$ is contained in at most 
$O(1)$ halo-sets, assuming that $\nterms\geq 2\ell$. 
Hence, we have to apply the rooted subset $(\ell+1)$-connectivity
algorithm $O(1)$ times.

We now analyze the number of micro iterations
needed to hit all of the halo-families.
Let $h_i$ denote the number of halo-families that are not hit by
$\setS$ at the beginning of the $i$-th micro iteration.
Recall that the number of cores after the preprocessing step is
$O(\ell)$. Thus, $h_1= O(\ell)$.
We claim that, at the $i$-th iteration, the number of halo-families
that are not hit by $\setS$ is at most $h_1/2^{i-1}$.

\begin{lemma}
Consider the $i$-th micro iteration. The number of halo-families that
are not hit by terminals of $\setS$ at the start of the iteration is
$h_1/2^{i-1}$. 
\label{lmm:uncovered-halo}
\end{lemma}

\begin{proof}
We proceed by induction on $i$. 
It is trivial for $i=1$. 
Suppose that the assertion is true for the $(i-1)$-th micro iteration
for some $i>1$.
Consider the $(i-1)$-th micro iteration.
Since we choose a vertex $\hr$ with a minimum thickness, by
Lemma~\ref{lmm:low-thickness}, the thickness of $\hr$ is at most
$\frac{h_{i-1}\ell}{\nterms}$.
Note that $\ell/\nterms\leq 1/2$ since $\nterms\geq 2\ell$. 
This means that $\hr$ is a neighbor of at most 
$h_{i-1}/2$ halo-sets.
At the end of the micro iteration,
halo-families that are not hit by terminals of $\setS$ are
halo-families whose halo-sets have $\hr$ as a neighbor.
Thus, the number of remaining halo-families is at most $h_{i-1}/2$. 
Hence, we have   
\[
h_i \leq h_{i-1}\left(\frac{1}{2}\right)
\leq \left(h_1\left(\frac{1}{2}\right)^{i-2}\right)
     \left(\frac{1}{2}\right)
= h_1\left(\frac{1}{2}\right)^{i-1} 
\]
\end{proof}

Lemma~\ref{lmm:uncovered-halo} implies that the maximum number of
micro iterations (within the covering-procedure) is
$O(\log{h_1})=O(\log{\ell})$.
So, in each inner iteration, we have to call the rooted subset
$(\ell+1)$-connectivity algorithm $O(\log{\ell})$ times.

Lastly, we analyze the total cost incurred by all inner iterations,
which is the cost for increasing the subset connectivity of a graph by
one.
We may apply Theorem~\ref{thm:rooted-kconn} directly to analyze the
cost of the solution. However, this leads to a bound slightly weaker
than what we claimed. To get the desired bound, we apply a stronger
version of Nutov's theorem~\cite{Nutov12}. In particular, the
approximation guarantee of Nutov's algorithm depends on the size of a
smallest deficient set. To be precise, {\em the size of a smallest
  deficient set} is defined by
\[
\min\{|U\cap\setT|: \mbox{$U$ is a deficient set}\}.
\]

\begin{lemma}[Nutov 2009~\cite{Nutov12}]
\label{lmm:root-size}
Consider the problem of increasing the rooted subset connectivity of a
graph from $\ell$ to $\ell+1$. 
Let $\phi=\min\{|U\cap\setT|:U \mbox{ is a deficient set}\}$. 
That is, each deficient set of the initial graph contains at least
$\phi$ terminals. 
Then there is an $O(\ell/\phi)$-approximation algorithm. 
\end{lemma}

Now, we analyze the size of a smallest deficient set of a graph at the 
beginning of each inner iteration. 
Consider the cores at any inner iteration. We call cores at the
beginning of the iteration {\em old cores} and call cores at the end
of the iteration {\em new cores}. We claim that every new core $\hC$
contains at least two old cores $C$ and $D$ that are disjoint on
$\setT$. This follows from the following lemma.

\begin{lemma}
No small deficient set contains two distinct cores $C$ and $D$
such that $C\cap D\cap \setT\neq\emptyset$.
\label{lmm:two-cores}
\end{lemma}

\begin{proof}
Suppose to the contrary that there is a small deficient set $U$ that
contains two distinct cores $C$ and $D$ such that $C\cap D\cap
\setT\neq\emptyset$.
Since $C\cap D\cap \setT\neq\emptyset$,
by Lemma~\ref{lmm:sskconn-disjoint}, 
$\inv{C}\cap\inv{D}\cap \setT=\emptyset$; 
that is, $\inv{C}\cap\inv{D}$ has no terminals.
Since $U$ contains both $C$ and $D$,
it follows that $\inv{U}$ is contained in both $\inv{C}$ and $\inv{D}$.
Hence, $\inv{U}$ has no terminals.
This contradicts the fact that $U$ is a deficient set.
\end{proof}

Lemma~\ref{lmm:two-cores} implies that no new cores contain two old
cores that are intersecting on $\setT$. This is because new cores
are small deficient sets of the old graph. Moreover, since all small 
deficient sets that contain only one core have been covered, new cores
must contain at least two old cores that are disjoint on $\setT$.
Thus, the size of a smallest deficient set increases by a factor of
$2$. This implies the following lemma. 

\begin{lemma}
Consider the $j$-th inner iteration. At the beginning of the
iteration, the size of a smallest deficient set of the current graph
is at least $2^{j-1}$.
\label{lmm:size-defi}
\end{lemma}

\begin{proof}
As in the above discussion, the size of a smallest deficient set
increases by a factor of two in each inner iteration. 
In more detail, consider the size of a smallest deficient set at the
beginning and the end of an inner iteration. 
We call the graph at the beginning of the iteration an ``old graph''
and the graph at the end of the iteration a ``new graph''. 
Let $U$ and $U'$ denote smallest deficient sets of the old and the
new graph, respectively.
By the size argument, we conclude that $U$ and $U'$ are cores of the 
old and the new graph.
At the end of the inner iteration, small deficient sets containing
one core are all covered. 
Thus, $U'$ contains two distinct cores $C$ and $C'$ of the old graph. 
Moreover, Lemma~\ref{lmm:two-cores} implies that $C$ and $C'$ have no
terminals in common. 
By the minimality of $U$, we have 
$|C\cap\setT|\geq|U\cap\setT|$ and $|C'\cap\setT|\geq|U\cap\setT|$. 
Thus, $|U'\cap T|\geq |C\cup\setT|+|C'\cap\setT|\geq 2|U\cap T|$ as
claimed. 

Now, we prove the lemma by induction. 
At the first inner iteration, each deficient set contains at least one
terminal.
Thus, the statement holds for the base case.
Assume that the assertion is true for the $(j-1)$-th inner iteration;
that is, at the beginning of the $(j-1)$-th inner iteration, any
deficient set $U$ has at least $2^{j-2}$ terminals.  
By the above claim, this number increases by a factor of two at the
end of the iteration. 
Thus, at the beginning of the $j$-th iteration, the size of a
smallest deficient sets is $2^{j-1}$, proving the lemma.
\end{proof}

By Lemma~\ref{lmm:root-size} and \ref{lmm:size-defi}, at the $j$-th
inner iteration, the cost incurred by the rooted subset
$(\ell+1)$-connectivity algorithm is $O(\ell/2^{j-1})$. 
Combining everything together, the approximation guarantee for the 
problem of increasing the subset connectivity of a graph by one is 
\[
O\left(\frac{\ell}{2^0}\log{\ell} + \frac{\ell}{2^1}\log{\ell} +
  \ldots \right)
= O(\ell\log{\ell}). 
\]

Thus, by Theorem~\ref{lmm:LP-scaling}, our algorithm achieves
an approximation guarantee of $O(k \log^2{k})$, assuming that
$\nterms\ge 2k$. 

\subsection{An $O(k\log{k})$-approximation algorithm for 
  $\nterms\geq k^2$.}
\label{sec:very-simple-algo}

To finish, we show that if the number of terminals is large, then we
get a slightly better performance guarantee. 
Observe that if $\nterms\ge k^2$, then, by
Lemma~\ref{lmm:low-thickness}, there is a terminal $\hr$ with a 
thickness of at most $\frac{q\ell}{\nterms} \leq
\frac{2\ell^2}{\ell^2} = 2$.
Moreover, by Lemma~\ref{lmm:bound-1}, each terminal is contained in at
most $\frac{2\nterms}{\nterms-\ell}=O(1)$ halo-sets.
Thus, the number of halo-families that are not hit by $\hr$ is $O(1)$.
This means that we can hit all the remaining halo-families by
choosing $O(1)$ terminals; that is, for each halo-family, we choose
one terminal from its core.
So, we can skip the micro iterations of the covering-procedure, and  
the approximation guarantee becomes $O(k\log{k})$.

\subsection{Analysis for the case $k<\nterms<2k$} 
\label{sec:below2k}

Our algorithm in Section~\ref{sec:second-algo} indeed applies to the
case $k<\nterms<2k$ with an approximation guarantee of
$O\left(\left(\frac{\nterms}{\nterms-k}\right)^2k\log^2{k}\right)$. 
To see this, we leave the bounds in Lemma~\ref{lmm:bound-1} and
Lemma~\ref{lmm:low-thickness} untouched. 
Then we have 
\begin{itemize} 
\item Each terminal is contained in at most
  $O\left(\frac{\nterms}{\nterms-k}\right)$ halo-families.
\item There is a terminal with a thickness of 
      $O\left(\frac{\ell q}{\nterms}\right)$, 
      where $q$ is the number of halo-families. 
\end{itemize}

Recall the micro iterations of the covering-procedure. 
In each micro iteration, we choose
$O\left(\frac{\nterms}{\nterms-\ell}\right)$ terminals, and 
the number of halo-families (which are not hit) decreases by a
factor of $\frac{\nterms}{\ell}$.
Here the number of micro iterations is not logarithmic
because $\frac{\nterms}{\ell}$ is not constant when
$\nterms\approx\ell$.
To analyze the upper bound, we write $\frac{\ell}{\nterms}$ as 
$1-\frac{1}{\nterms/(\nterms-\ell)}$ and apply an equation:
\[
\lim_{x\rightarrow\infty}\left(1-\frac{1}{x}\right)^x=\frac{1}{e}
\]
Thus, we need $O\left(\frac{\nterms}{\nterms-\ell}\right)$ micro
iterations to decrease the number of halo-families (which are not
hit) by a factor of $e$.
This means that the covering-procedure terminates in
$O\left(\frac{\nterms}{\nterms-\ell}\log{q}\right)$ iterations, where
$q$ is the number of halo-families.
(Note that, in this case, we do not need the preprocessing step because
the number of halo-families is at most $\nterms^2=O(k^2)$.)
So, the covering-procedure has to call the rooted subset
$(\ell+1)$-connectivity algorithm for 
$O\left(\left(\frac{\nterms}{\nterms-\ell}\right)^2\log{k}\right)$
times.
Following the analysis in Section~\ref{sec:analysis}, we have an
approximation guarantee of
$O\left(\left(\frac{\nterms}{\nterms-k}\right)^2k\log^2{k}\right)$ as
claimed. 


\section{Hardness of  the subset $k$-connectivity problem}  
\label{sec:hardness}

In this section, we discuss the hardness of the subset
$k$-connectivity problem.  
First, we will show in Section~\ref{sec:hardest-instance} that the
hardest instance of the subset $k$-connectivity problem might be when
$k\approx |T|$; that is, we prove Proposition~\ref{prop:hardest}. 
Then we will present in Section~\ref{sec:root2subset} an approximation
preserving reduction from the rooted subset $k$-connectivity problem
to the subset $k$-connectivity problem; that is, we prove 
Theorem~\ref{thm:rooted-to-subset}.

\subsection{The hardest instance}
\label{sec:hardest-instance}

We will show that an $\alpha(k)$-approximation algorithm for the case
$\nterms=k$ implies an $(\alpha(k)+f(k))$-approximation algorithm for
all instances with $\nterms>k$, where $f(k)$ is the best known
approximation guarantee for the rooted subset $k$-connectivity
problem.
In particular, instances with $\nterms\approx{k}$ might be the hardest
cases of the subset $k$-connectivity problem.

%
Suppose there is an $\alpha(k)$-approximation algorithm $\calA$ for the
subset $k$-connectivity problem for the case $\nterms=k$. 
We apply $\calA$ to solve an instance of the subset $k$-connectivity
problem with $\nterms>k$ as follows.
Let $G=(V,E)$ be a given graph and $\setT\subseteq V$ be a set of
terminals, where $\nterms>k$.
First, we take any subset $\setR$ of $k$ terminals from $\setT$.
Then we apply the algorithm $\calA$ to this instance with $\setR$ as the
set of terminals; this results in a graph $G_{\setR}=(V,E_{\setR})$. 
Clearly, $\setR$ is $k$-connected in $G_{\setR}$. 
Now, we make the remaining terminals connected to $\setR$ by applying
the rooted subset $k$-connectivity algorithm.
To be precise, we construct a padded graph by adding a new vertex $\hr$
and new edges of zero cost from $\hr$ to each terminal of
$\setR$. Then we apply the rooted subset $k$-connectivity algorithm
to the padded graph with the set of terminals $\setT$ and the root
$\hr$. 
Denote a solution subgraph (of the padded graph) by 
$G_{pad} = (V\cup\{\hr\},\,E_{root}\cup\{(\hr,t):t\in \setR\})$, where 
$E_{root}\subseteq E$.
The algorithm outputs the union of the two subgraphs, namely 
$\hG=(V,E_{\setR}\cup E_{root})$.
 
We claim that the set of all terminals $T$ is $k$-connected in $\hG$.
Suppose not.
Then there is a set of vertices $X\subseteq V$ of size $k-1$ that
separates some terminals $s,t\in\setT\setminus X$; that is, $s$ and
$t$ are not connected in $\hG\setminus X$.
Consider the padded subgraph $G_{pad}$. 
By the construction, since $G_{pad}$ is $k$-connected from $\hr$ to
$\setT$, both $s$ and $t$ have paths to $\hr$ in
$G_{pad}\setminus{X}$. 
Moreover, each of these two paths must visit some terminals $s'$ and
$t'$ in $\setR$,  respectively. 
If $s'=t'$, then $s$ and $t$ are connected by the union of these
paths. So, we have a contradiction. 
If $s'\neq t'$, then we can join these two paths by an $s',t'$-path in
$G_{\setR}\setminus X$. 
Such $s',t'$-path exists because $\setR$ is $k$-connected in
$G_{\setR}$, meaning that $X$ cannot separates a pair of terminals in 
$\setR$.  
Thus, $s$ and $t$ are connected, and we again have a contradiction.

Now, consider the cost. The approximation factor incurred by the
algorithm $\calA$ is $\alpha(k)$, and the approximation factor
incurred by the rooted subset $k$-connectivity algorithm is
$f(k)$. Thus, the above algorithm gives an approximation guarantee
of $(\alpha(k)+f(k))$ as claimed.

\subsection{A reduction from the rooted subset $k$-connectivity problem}
\label{sec:root2subset}

As we showed in the previous section, an approximation algorithm
for the rooted subset $k$-connectivity problem implies an
approximation algorithm for the subset $k$-connectivity problem. 
Hence, it is more likely that the rooted problem is easier than the
subset problem. 
Here we show a solid evidence of this statement; that is, we will give
an approximation preserving reduction from the rooted subset 
$k$-connectivity problem to the subset $k$-connectivity problem. 

The key idea of the reduction is that a solution $\hG$ to the rooted
subset $k$-connectivity problem is indeed almost subset
$k$-connected. 
In particular, if the root vertex $r$ is not allowed to be removed,
then there is no set of vertices of size less than $k$ that can
separate a pair of terminals. 
So, we want to prevent the root vertex $r$ from being in a
separator.
To do this, we replace $r$ by a clique $K_d$ of size at least
$k+1$. 
Thus, removing any set of less than $k$ vertices cannot remove all
vertices corresponding to $r$.

Now, we shall realize the above idea. 
First, take any instance $\Pi^{root}$ of the rooted subset
$k$-connectivity problem consisting of a graph $G=(V,E)$, a set of
terminals $\setT\subseteq V$ and a root vertex $r\in
V\setminus\setT$. 
Let $d$ be the degree of $r$ in $G$.
Clearly, if the instance $\Pi^{root}$ is feasible, then $d\geq k$.  
We construct an instance $\Pi^{subset}$ of the subset $k$-connectivity
problem consisting of a graph $G'=(V',E')$ and a set of terminals
$\setT'$ as follows. 
%
Let $\{v_1,v_2,\ldots,v_d\}$ be a set of neighbors of $r$ in $G$. 
We remove from $G$ the vertex $r$ and replace it with a clique
$K_{d+1}$ on a set of vertices $\{r',v'_1,v'_2,\ldots,v'_d\}$.
All edges of $K_{d+1}$ have zero costs. 
The vertex $r'$ corresponds to the root vertex $r$ of $G$, and each
vertex $v'_i$ corresponds to each neighbor $v_i$ of $r$ in $G$. 
Then we connect $K_{d+1}$ to $G$ by adding to $G'$ an edge
$(v'_i,v_i)$ for each edge $(r,v_i)$ in $G$ and setting the cost of
$(v'_i,v_i)$ to be the same as the cost of $(r,v_i)$. 
Thus, each edge $(v'_i,v_i)$ in $G'$ corresponds to an edge $(r,v_i)$
in $G$. 
The set of terminals of this new instance is $\setT'=\setT\cup\{r'\}$,
and the connectivity requirements is $k$, the same for both instances. 
This completes the construction.

In sum, we have 
\begin{align*}
   G' &= G - \{r\} + K_{d+1} + 
           \{(v'_1,v_1),(v'_2,v_2),\ldots,(v'_d,v_d)\}\\
   T' &= T\cup\{r'\}\\
   r' &\leftrightarrow r\\
   (v'_i,v_i) &\leftrightarrow (r,v_i)\mbox{ for all $i=1,2,\ldots,d$}\\
\end{align*}


\noindent{\bf Completeness:}
First, we show that any feasible solution $H$ of $\Pi^{root}$ maps to 
a feasible solution $H'$ of $\Pi^{subset}$ with the same cost.
The mapping is as follows.
Given a graph $H$, we construct a solution $H'$ to $\Pi^{subset}$ by 
taking all edges of $K_{d+1}$ and all edges of $G'$ corresponding to
edges of $H$.  
Clearly, the cost of $H'$ and $H$ are the same. 
It remains to show that $T'$ is $k$-connected in $H'$. 

The connectivity between the vertex $r'$ and each terminal $t\in T$ is
clearly satisfied. 
This is because any collection of openly disjoint $r,t$-paths in
$H$ maps to a collection of openly disjoint $r',t$-paths in $H'$. 
In particular, any path $P=(r,v_i,\ldots,t)$ in $H$ maps to a path
$P'=(r,v'_i,v_i,\ldots,t)$ in $H'$, and it is easy to see that the
mapping preserves vertex-disjointness. 
By the same argument, we can deduce that every vertex $v'_j\in
K_{d+1}$ is $k$-connected to $t$ in $H'$.
This is because the path $P$ also maps to a path
$P''=(v'_j,v'_i,v_i,\ldots,t)$ or $P''=(v'_j,v_j,\ldots,t)$ in $H'$.  

Now, consider the connectivity between a pair of vertices
$t,t'\in\setT$. 
Assume a contradiction that $t$ and $t'$ are not $k$-connected.
Then there is a subset of vertices $X$ of $G'$ with $|X|\leq k-1$ such
that $t$ and $t'$ are not connected in $H'\setminus{X}$.  
Since $|X| \leq k-1 \leq d$, there is a vertex $s$ in 
$K_{d+1}\setminus{X}$. 
(The vertex $s$ is either the vertex $r'$ or some vertex $v_i$ in
$K_{d+1}$.)
As we have shown, $s$ is $k$-connected to $t$ and $t'$ in $H'$.
Thus, by Menger's theorem, $H'\setminus{X}$ contains both an $s,t$-path
and an $s,t'$-path.
So, $t$ and $t'$ are connected in $H'\setminus{X}$, a
contradiction.  
Therefore, $\setT'$ is $k$-connected in $H'$, implying that $H'$ is
feasible to the subset $k$-connectivity problem. 


\medskip

\noindent{\bf Soundness:} Now, we show the converse; that is, 
any feasible solution $H'$ of $\Pi^{subset}$ maps to a feasible
solution $H$ of $\Pi^{root}$ with the same cost. 
This direction is easy. 
We construct $H'$ by taking all edges of $H'$ that correspond to edges
of $G$. Clearly, the cost of $H$ and $H'$ are the same. 
By feasibility, $H'$ has, for each terminal $t\in T$,  a collection of
$k$ openly disjoint $r',t$-paths, namely $P'_1,P'_2,\ldots,P'_k$, and 
each path $P'_j$ is of the form $P'_j=(r',v'_i,v_i,\ldots,t)$. 
The path $P'_j$ maps to a path $P_j=(r,v_i,\ldots,t)$ in $H$. 
So, we have a collection of paths $P_1,P_2,\ldots,P_k$ in $H$ that are
openly disjoint. 
Therefore, $H$ is feasible to the rooted subset
$k$-connectivity problem, finishing the proof.


\section{Conclusions and Discussions}
\label{sec:discuss}
\label{sec:conclusion}

We studied the structure of the subset $k$-connectivity problem and
used this knowledge to design an approximation algorithm for the
subset $k$-connectivity problem.
When the number of terminals is moderately
large, at least $2k$, our algorithm gives a very good approximation
guarantee of $O(k\log^2k)$. 
When the number of terminals is tiny, at most $\sqrt{k}$, then the
trivial algorithm also gives a very good approximation guarantee of
$O(k)$.
However, when the number of terminals is between $\sqrt{k}$ and $2k$,  
the approximation guarantee can be as large as $\Theta(k^2)$.
Interestingly, as we have shown, it does seem that the hardest
instances of the subset $k$-connectivity problem are when the number
of terminals is close to $k$.

\bigskip

\noindent {\bf Acknowledgments.}
We thank Joseph Cheriyan for useful discussions over a year. 
Also, we thank Adrian Vetta, Parinya Chalermsook, Danupon Nanongkai,
Jittat Fakcharoenphol and anonymous referees for useful comments on
the preliminary draft.


\bibliographystyle{plain}
\bibliography{subset-k-conn}


\end{document}